\documentclass[a4paper,11pt]{article}
\usepackage[utf8]{inputenc}
\usepackage{amsmath,amsthm,amssymb}
\usepackage{todonotes}

\usepackage{tikz-qtree}
\usepackage{tikz}
\usetikzlibrary{positioning, backgrounds,fit}
\usetikzlibrary{shapes,arrows}

\newtheorem{theorem}{Theorem}
\newtheorem{observation}[theorem]{Observation}
\newtheorem{lemma}[theorem]{Lemma}
\newtheorem{corollary}[theorem]{Corollary}
\newtheorem{proposition}[theorem]{Proposition}

\title{Knowledge Compilation, Width and Quantification \thanks{This work was partially supported by the French Agence Nationale de la Recherche, AGGREG project reference ANR-14-CE25-0017-01.}}
\author{Florent Capelli\thanks{florent.capelli@univ-lille.fr, Université de Lille, CRIStAL, CNRS/Inria} \and Stefan Mengel\thanks{mengel@cril.fr, CNRS, CRIL UMR 8188}}

\newcommand{\var}{\mathsf{var}}
\newcommand{\forgot}{\mathsf{forgot}}
\newcommand{\kept}{\mathsf{kept}}
\newcommand{\shape}{\mathsf{Shape}}

\newcommand{\join}{\bowtie}

\newcommand{\calC}{\mathcal{C}}

\begin{document}

\maketitle

\begin{abstract}
  We generalize many results concerning the tractability of SAT and \#SAT on bounded treewidth CNF-formula in the context of Quantified Boolean Formulas (QBF). To this end, we start by studying the notion of width for OBDD and observe that the blow up in size while  existentially or universally projecting a block of variables in an OBDD only affects its width. We then generalize this notion of width to the more general representation of structured (deterministic) DNNF and give a similar algorithm to existentially or universally project a block of variables. Using a well-known algorithm transforming bounded treewidth CNF formula into deterministic DNNF, we are able to generalize this connection to quantified CNF which gives us as a byproduct that one can count the number of models of a bounded treewidth and bounded quantifier alternation quantified CNF in FPT time.  We also give an extensive study of bounded width d-DNNF and proves the optimality of several of our results.
\end{abstract}

\section{Introduction}

It is well known that restricting the interaction between variables
and clauses in CNF-formulas makes several hard problems on them tractable. For
example, the propositional satisfiability problem SAT and its counting version 
\#SAT can be solved in time $2^{O(k)}|F|$ when $F$ is a
CNF formula whose primal graph is of treewidth
$k$~\cite{Szeider04,SamerS10}. Many extensions of this result have
been shown these last ten years for more general graph
measures such as modular treewidth or cliquewidth~\cite{PaulusmaSlivovskySzeider16,SlivovskyS13,SaetherTV14}. In~\cite{BovaCMS15},
Bova et al.~recently explained these results using Knowledge
Compilation, a subarea of artificial intelligence that systematically studies and compares the properties of different representations for knowledge: many classes of structured CNF can be represented by
small Boolean circuits known as structured deterministic DNNF~\cite{PipatsrisawatD08}. 
Such circuits have strong restrictions making several
problems such as satisfiability and model counting on them tractable.

In this paper, we show how these circuit representations can be
used in the context of quantification. To this end, we give a simple
algorithm that, given a structured d-DNNF $D$ and a subset $Z$ of variables,
outputs a structured d-DNNF $D'$ computing $\exists Z\;D$. We show a similar
result to construct a structured d-DNNF $D'$ computing $\neg \exists Z\;D$. In
general, the size of $D'$ blows up exponentially during our
transformation and this is unavoidable since there are strong
exponential lower bounds in the setting~\cite{PipatsrisawatD08}. But here we define a notion of width for complete structured d-DNNF that generalizes more well-known notions like width of OBDD or SDD and show that the
exponential blowup in fact depends only on the \emph{width} of the
input circuit and not on the size. Since many structured CNF-formulas,
such as those of bounded treewidth, can be translated into complete 
structured d-DNNF of bounded width, we are able to construct structured d-DNNF for the quantified formula where the blowup is
relatively tame in our setting which yields fixed-parameter tractable
algorithms for several problems. Figure~\ref{fig:scheme} depicts the
overall scheme that we use to construct such algorithms.

\tikzstyle{block} = [rectangle, draw, fill=gray!20, 
     text width=11em, minimum height=4em]
\tikzstyle{line} = [draw, -latex']

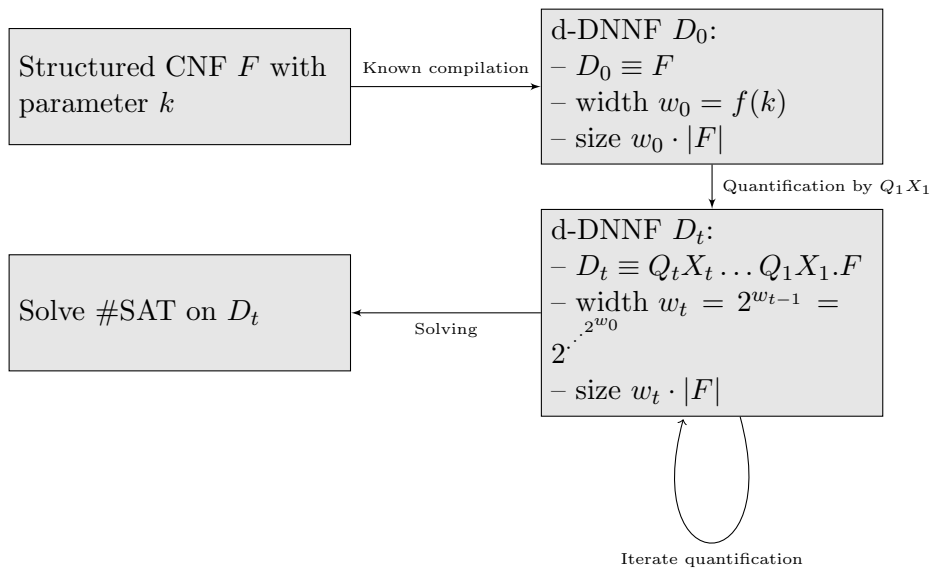
\begin{figure*}
  \centering
  \begin{tikzpicture}[node distance = 7cm, auto, scale=0.5]
    \node [block] (init) {Structured CNF $F$ with parameter $k$};
    \node [block, right of=init] (compile) {
      d-DNNF $D_0$:\\
      -- $D_0 \equiv F$ \\
      -- width $w_0 = f(k)$ \\
      -- size $w_0 \cdot |F|$ 
    };
    \node [block, below of=compile, node distance = 3cm] (quantify) {
      d-DNNF $D_t$:    \\
      -- $D_t \equiv Q_t X_t \dots Q_1 X_1.F$\\
      -- width $w_t = 2^{w_{t-1}} = 2^{\cdot^{\cdot^{\cdot^{2^{w_0}}}}}$        \\
      -- size $w_t \cdot |F|$   
    };

    \node [block, left of=quantify] (final) {
      Solve \#SAT on $D_t$
      };
      \path [line] (init) -- node  {\tiny Known compilation} (compile);
      \path [line] (compile) -- node {\tiny Quantification by $Q_1X_1$}(quantify);
      \path [line] (quantify) edge [loop below]  node {\tiny Iterate quantification} (quantify);
      \path [line] (quantify) -- node {\tiny Solving} (final);
\end{tikzpicture}

\caption{The overall scheme for proving tractability results on structured quantified CNF.}
\label{fig:scheme}
\end{figure*}
For instance, our algorithm can be used to show that the number of
models of a partially quantified CNF-formula $F$ of treewidth $k$ with
$t$ blocks of quantifiers can be computed in time
$2^{\cdot^{\cdot^{\cdot^{2^{O(k)}}}}} |F|$ with $t+1$
exponentiations. This generalizes a result by Chen~\cite{Chen04} where
the fixed-parameter tractability of QBF on such formulas was shown
with a comparable complexity. Moreover, it generalizes a very recent result of~\cite{FichteMHW18} on projected model counting, i.e., model counting in the presence of a single existential variable block. Finally, our algorithm also applies to the more general notions of incidence treewidth and signed cliquewidth.

We complement our algorithm with lower bounds that show that our construction 
is essentially optimal in several respects.

The paper is organized as follows: Section~\ref{sec:preliminaries} introduces the necessary preliminaries. Section~\ref{sec:warmup} is dedicated to show how quantification can be efficiently done on small width complete OBDD. The aim of this section is to present our results in a simpler framework. Section~\ref{sec:ddnnf} generalizes the result of Section~\ref{sec:warmup} to the more powerful representation of bounded width d-DNNF. The rest of the paper is dedicated to corollaries of this result proven in Section~\ref{sec:ddnnf} and explores the limits and optimality of our approach. Section~\ref{sec:graphs} is dedicated to prove parametrized tractability results for QBF when the graph of the input CNF is restricted. Section~\ref{sec:transformations} gives a systematic study of the tractable transformations of bounded width d-DNNF, in the spirit of~\cite{DarwicheM02}. Finally, Section~\ref{sec:lower-bounds} contains several results showing that our definition of bounded width DNNF cannot be straightforwardly weakened while still supporting efficient quantification.

\section{Preliminaries}
\label{sec:preliminaries}

By $\exp^\ell(p)$ we denote the iterated exponentiation function that is defined by $\exp^0(p):= p$ and $\exp^{\ell+1}(p) := 2^{\exp^\ell(p)}$.

\paragraph{CNF and QBF.} We assume that the reader is familiar with the basics of Boolean logic and fix some notation. For a Boolean function $F$ and a partial assignment $\tau$ to the variables of $F$, denote by $F[\tau]$ the function we get from $F$ by fixing the variables of $\tau$ according to $\tau$.
A \emph{literal} is a Boolean variable or its negation. A \emph{clause} is a disjunction of literals and finally a \emph{formula in negation normal form} (short \emph{CNF formula}) is a conjunction of clauses. We define the size $C$ of a clause $C$ as the number of literals appearing in it. The size $|F|$ of a formula $F$ is then defined as $\sum_C|C|$ where the sum is over the clauses in $F$.

A \emph{ Quantified Boolean Formula} (short \emph{QBF}) $F = Q_1 X_1 Q_2 X_2 \dots Q_{\ell} X_\ell~F'$ is a CNF formula $F'$ together with a \emph{quantified prefix} $Q_1 X_1 Q_2 X_2 \dots \exists X_\ell$ where $X_1, \dots, X_\ell$ are disjoint subsets of variables of $F'$, $Q_i$ is either $\exists$ or $\forall$ and $Q_{i+1} \neq Q_i$. The number of blocks $\ell$ is called the \emph{quantifier alternation}. W.l.o.g, we always assume that $Q_\ell$, the most nested quantifier, is always an $\exists$-quantifier. The \emph{quantified variables} of $F$ are defined as $\bigcup_{i=1}^\ell X_i$ and the \emph{free variables} of $F$ are the variables of $F$ that are not quantified. A quantified CNF naturally induces a Boolean function on its free variables. 

\paragraph{Representations of Boolean functions.} We present several representations studied in the area of knowledge compilation in a rather succinct fashion. For more details and discussion, the interested reader is refered to~\cite{DarwicheM02,PipatsrisawatD08}.

A Boolean circuit $C$ is defined to be in \emph{negation normal form} (short an NNF) if $\neg$-gates appear in it only directly above the inputs. An $\land$-gate in an NNF is called \emph{decomposable} if, for its inputs $g_1, g_2$ the subcircuits rooted in $g_1$ and $g_2$ are on disjoint variable sets. A circuit in \emph{decomposable negation normal form} (short a DNNF) is an NNF in which all gates are decomposable~\cite{Darwiche01}. An $\lor$-gate $g$ in an NNF is called \emph{deterministic} if there is no assignment to the variables of the circuit that makes two children of $g$ true. A DNNF is said to be  \emph{deterministic} (short a \emph{d-DNNF}) if all its $\lor$-gates are deterministic.

A \emph{binary decision diagram} (short \emph{BDD}) is a directed acycliyc graph with the following properties: there is one source and two sinks, one of each labeled with $0$ and $1$. The non-sink nodes are labeled with Boolean variables and have two outgoing edges each, one labeled with $0$ the other with $1$. A BDD $B$ computes a function as follows: for every assignment $a$ to the variables of $B$, one constructs a source-sink path by starting in the source and in every node labeled with a Boolean variable $X$ following the edge labeled with $a(X)$. The label of the sink reached this way is then the value computed by $B$ on $a$.

A BDD is called a \emph{free} BDD (short \emph{FBDD}) if on every source-sink path every variable appears at most once. If on every path the variables are seen in a fixed order $\pi$, then the FBDD is called an ordered BDD (short \emph{OBDD}).

An FBDD is called \emph{complete} if on every source-sink path every variable appears exactly once. This notion also applies to OBDDs in the obvious way.
A \emph{layer} of a variable $X$ in a complete OBDD $B$ is the set of all nodes labeled with $X$. The \emph{width} of $B$ is the maximum size of its layers.
Note that for every OBDD one can construct a complete OBDD computing the same function in polynomial time, but it is known that it is in general unavoidable to increase the number of nodes labeled by a variable by a factor linear in the number of variables~\cite{BolligW00}.

\paragraph{Graphs of CNF formulas.}
There are two graphs commonly assigned to CNF formulas: the \emph{primal graph} of a CNF formula $F$ is the graph that has as its vertices the variables of $F$ and there is an edge between two vertices $x,y$ if and only if there is clause in $F$ that contains both $x$ and $y$. The \emph{incidence graph} of $F$ has as vertices the variables and the clauses of $F$ and there is an edge between two nodes $x$ and $C$ if and only if $x$ is a variable, $C$ is a clause, and $x$ appears in $C$.

We will consider several width measures on graphs like treewidth and pathwidth. Since we do not actually need the definitions of these measures but only depend on known results on them, we spare the readers these rather technical definitions and give pointers to the literature in the respective places.



\section{Warm-up: Quantification on OBDD}
\label{sec:warmup}

In this section, we will illustrate the main ideas of our approach on the simpler case of OBDD. To this end, fix an OBDD $G$ in variables $X_1, \ldots, X_n$ in that order. Now let $Z$ be a set of variables. We want to compute an OBDD that encodes $\exists Z\, G$, i.e., we want to forget the variables in $Z$.

Note that it is well-known that OBDDs do not allow arbitrary forgetting of variables without an exponential blow-up, see~\cite{DarwicheM02}. Here we make the observation that this exponential blow-up is in fact not in the \emph{size} of the considered OBDD but in the \emph{width} which for many interesting cases is far lower.

\begin{lemma}\label{lem:OBDD}
 Let $G$ be a complete OBDD of width $w$ and $Z$ be a set of its variables. Then there is an OBDD is width $2^w$ that computes the function of $\exists Z\; G$.
\end{lemma}
\begin{proof}
 The technique is essentially the power set construction used in the determinization of finite automata. Let $V_x$ for a variable $x$ denote the set of nodes labeled by~$x$. For every $x$ not in $Z$, our new OBDD $G'$ will have a node $N_{S,x}$ labeled by $x$ for every subset $S\subseteq V_x$. The invariant during the construction will be that a partial assignment $a$ to the variables in $\var(G)\setminus Z$ that come before $x$ in $G$ leads to $N_{S,x}$ if and only if $S$ is the set of nodes in $V_x$ which can be reached from the source by an extension of $a$ on the variables of $Z$. We make the same construction for the $0$- and $1$-sink of $G$: $G'$ gets three sinks $0$, $1$ and $01$ which encode which sinks of $G$ can be reached with extensions of an assignment $a$. Note that if we can construct such a $G'$, we are done by merging the sinks $1$ and $01$.
 
 The construction of $G'$ is fairly straightforward: for every variable $x$ not in $Z$, for every node $N\in V_x$, we compute the set of nodes $N^+$ labeled with the next variable $x'$ not in $Z$ that we can reach by following the $1$-edge of $N$ and the set of nodes $N^-$ we can reach by following the $0$-edge of $N$. Then, for every $S\subseteq V_x$ we define the $1$-successor of $N_{S,x}$ as $N_{S', x'}$ where $S'= \bigcup_{N\in S}N^+$. The $0$-successors are defined analogously.
\end{proof}

We remark that in~\cite{FerraraPV05} a related result is shown: for a CNF-formula~$F$ of pathwidth $k$ and every subset $Z$ of variables, one can construct an OBDD of size $2^{2^k}|F|$ computing $\exists Z\; F$. This result follows easily from Lemma~\ref{lem:OBDD} by noting that for a CNF $F$ of pathwidth $k$ one can construct a complete OBDD of width $2^p$. We note that our approach is more flexible than the result in~\cite{FerraraPV05} because we can iteratively add more quantifier blocks since $\forall Z \; D \equiv \neg (\exists Z \neg D)$ and negation in OBDD can easily easily performed without size increase. For example, one easily gets the following corollary.

\begin{corollary}\label{cor:weakHubie}
 There is an algorithm that, given a QBF restricted to $\ell$ quantifier alternations and of pathwidth $k$, decides if $F$ is true in time $O(\exp^\ell(p)|F|)$.
\end{corollary}
Note that Corollary~\ref{cor:weakHubie} is already known as it is a special case of the corresponding result for treewidth in~\cite{Chen04}. However, we will show that a similar approach to that of Lemma~\ref{lem:OBDD} can be used to derive several generalizations of the result of~\cite{Chen04}: we show that we can add quantification to bounded width structured d-DNNF, a generalization of OBDD (see Section~\ref{sec:ddnnf}). Since several classes of CNF formulas are known to yield bounded width structured d-DNNF~\cite{BovaCMS15}, this directly yields QBF algorithms for these classes, see Section~\ref{sec:graphs} for details.

\section{Bounded width structured d-DNNF}
\label{sec:ddnnf}

\subsection{Definitions}
\label{sec:defnorm}

\paragraph{Complete structured DNNF.} A {\em vtree} $T$ for a set of variables $X$ is a rooted tree where every non-leaf node has exactly two children and the leaves of $T$ are in one-to-one correspondence with $X$. A {\em complete structured} DNNF $(D,T,\lambda)$ is a DNNF $D$ together with a vtree $T$ for $\var(D)$ and a labelling $\lambda$ of the nodes of $T$ with gates of $D$ such that:
\begin{itemize}
\item If $t$ is a leaf of $T$ labeled with variable $x \in X$ then $\lambda(t)$ contains only inputs of $D$ labeled with either $x$, $\neg x$.
\item For every gate $u$ of $D$, there exists a unique node $t_u$ of $T$ such that $u \in \lambda(t)$.
\item There is no non-leaf node $t$ of $T$ such that $\lambda(t)$ contains an input of $D$.
\item For every $\land$-gate $u$ with inputs $v_1, v_2$, we have $t_{v_1} \ne t_{v_2}$.
\item For every edge $(u,v)$ of $D$:
  \begin{itemize}
  \item Either $v$ is an $\wedge$-gate, $u$ is an $\vee$-gate or an input and $t_u$ is the child of $t_v$.
  \item Or $v$ is an $\vee$-gate, $u$ is an $\wedge$-gate and $t_u = t_v$.
  \end{itemize}
\end{itemize}
Intuitively, $T$ can be seen as a skeleton supporting the gates of $D$, as depicted on Figure~\ref{fig:sdnnfex}. In the following, when the vtree and its labelling is not necessary, we may refer to a complete structured DNNF $(D,T,\lambda)$ by only mentioning the circuit $D$. 

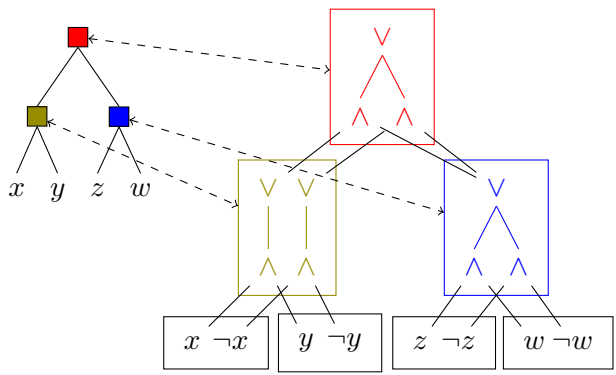
\begin{figure}
  \centering
  \begin{tikzpicture}
  \begin{scope}
  \Tree [.\node[draw, fill=red](s){~}; [.\node[draw, fill=olive](t){~}; $x$ $y$ ] [.\node[draw, fill=blue](u){~}; $z$ $w$ ] ];
\end{scope}

\begin{scope}[shift={(4cm, 0)}]
  
\tikzset{every tree node/.style={color=red}}
\tikzset{edge from parent/.style={draw, color=red}}

\Tree [.\node(r){$\vee$}; \node(a1){$\wedge$}; \node(a2){$\wedge$}; ];

\node[draw,  fit=(r) (a1) (a2), color=red] (redbox) {};
\draw[<->, dashed] (s) -- (redbox);

\begin{scope}[shift={(-1.5cm,-2cm)}]
\tikzset{every tree node/.style={color=olive}}
\tikzset{edge from parent/.style={draw, color=olive}}

\Tree [.\node(v1){$\vee$}; \node(l1){$\wedge$}; ];
\end{scope}

\begin{scope}[shift={(-1cm,-2cm)}]
\tikzset{every tree node/.style={color=olive}}
\tikzset{edge from parent/.style={draw, color=olive}}

\Tree [.\node(v2){$\vee$}; \node(l2){$\wedge$}; ];

\node[draw,  fit=(v1) (l1) (l2) (v2), color=olive] (olivebox) {};

\draw[<->, dashed] (t) -- (olivebox);
\end{scope}

\begin{scope}[shift={(1.5cm,-2cm)}]
\tikzset{every tree node/.style={color=blue}}
\tikzset{edge from parent/.style={draw, color=blue}}

\Tree [.\node(v3){$\vee$}; \node(l3){$\wedge$}; \node(l4){$\wedge$}; ];

\node[draw,  fit=(v3) (l3) (l4), color=blue] (bluebox) {};
\draw[<->, dashed] (u) -- (bluebox);

\end{scope}

\draw (a1) -- (v1);
\draw (a2) -- (v3);
\draw (a1) -- (v3);
\draw (a2) -- (v2);

\node[below of=l1, left of=l1] (x) {$x$};
\node[right of=x, node distance = 0.5cm] (nx) {$\neg x$};

\node[right of=nx, node distance=1cm] (y) {$y$};
\node[right of=y, node distance = 0.5cm] (ny) {$\neg y$};

\node[right of=ny] (z) {$z$};
\node[right of=z, node distance = 0.5cm] (nz) {$\neg z$};

\node[right of=nz] (w) {$w$};
\node[right of=w, node distance = 0.5cm] (nw) {$\neg w$};

\draw (y) -- (l1) -- (x);
\draw (ny) -- (l2) -- (nx);

\draw (z) -- (l3) -- (w);
\draw (nz) -- (l4) -- (nw);

\node[draw,fit=(x) (nx)] {};
\node[draw,fit=(y) (ny)] {};
\node[draw,fit=(z) (nz)] {};
\node[draw,fit=(w) (nw)] {};

\end{scope}
  \end{tikzpicture}

  \caption{A vtree $T$ and a complete structured DNNF $(D,T,\lambda)$, where $\lambda$ is represented with colors and dashed arrows.}
  \label{fig:sdnnfex}
\end{figure}

\paragraph{Width.} The \emph{width} of a complete structured DNNF $(D,T,\lambda)$ is defined as $\max_{t \in V(T)} |\{ v \in \lambda(t) \mid v \text{ is an $\vee$-gate} \}|$. For example, the DNNF pictured on Figure~\ref{fig:sdnnfex} has width $2$ since the green node is labeled with $2$ $\vee$-gates. 

Note that for the width we do not take into account $\land$-gates. This is for several reasons: first, only considering $\lor$-gates simplifies some of the arguments later on and gives cleaner results and proofs. Moreover, it is not hard to see that when rewriting OBDD as DNNF, the width of the original OBDD is exactly the width of the resulting circuit. The same is also true for the width of SDD~\cite{BovaS17}, another important representation of Boolean function~\cite{Darwiche11}. Thus, width defined only on $\lor$-gates allows a tighter connection to the literature. Finally, the number of $\land$-gates in a complete structured DNNF is highly connected to the width as we define it as we see in the following observation.

\begin{observation}
  \label{obs:smallbags}
 Let $(D,T,\lambda)$ be a complete structured DNNF of width $w$. We can in linear time in $|D|$ compute a complete structured DNNF $(D',T,\lambda')$ of width $w$ and equivalent to $D$. Moreover, for every node $t$ of $T$, we have $|\lambda'(t)| \le (w^2+w)$. Observe that $D'$ is thus of size at most $2(w+w^2)|\var(D)|$.
\end{observation}
\begin{proof}
For the first statement, note that by definition there are at most $w$ $\vee$-gates in $\lambda(t)$. Now, the inputs of every $\wedge$-gates of $\lambda(t)$ are $\vee$-gates of $\lambda(t_1)$ and $\lambda(t_2)$ where $t_1,t_2$ are the children of $t$ in $T$. Thus, there are at most $w^2$ possible ways of branching these $\wedge$-gates. So if we eliminate $\land$-gates that have identical inputs and keep for every combination at most one of them, we get $D'$ with the desired size bound on $\lambda'(t)$. However, we can neither naively compare the children of all $\land$-gates nor order the $\land$-gates by their children to eliminate $\land$-gates with identical inputs since both approaches would violate the linear time requirement.

To avoid this slight complication, we proceed as follows: in a first step, we count the $\land$-gates in $\lambda(t)$. If there are at most $w^2$ of them, we satisfy the required upper bound, so we do nothing. Otherwise, we create a array of size $w^2$ indexed by the pairs of potential inputs of $\land$-gates in $\lambda(t)$. We initialize all cells to some null-value. Now we iterate over the $\land$-gates in $\lambda(t)$ and do the following for every such gate $u$: if the cell indexed by the children of $u$ is empty, we store $u$ in that cell and continue. If there is already a gate $u'$ in the cell, we connect all gates that $u$ feeds into to $u'$ and delete $u$ afterwards. It is easy to see that the resulting algorithm runs in linear time, computes a $D'$ equivalent to $D$ and satisfies the size bounds on $\lambda(t)$.

Since $T$ is a tree where every node but the leaves has exactly $2$ children, the number of nodes in $T$ is at most $2n$. Now, because of $|\lambda'(t)| \leq w^2+w$, the bound on $|D'|$ follows directly.
\end{proof}

We remark that complete structured  DNNF as defined above are more restrictive than structured DNNF as defined in~\cite{PipatsrisawatD08}. That definition only gives a condition on the way decomposable $\wedge$-gates can partition variables, following the vtree. However, it is not hard to see that one can add dummy gates ($\vee$-gate and $\wedge$-gate of fan-in one) to force the circuit to the form we define with only a polynomial increase in its size. However, this transformation may lead to large width circuits. Moreover, it follows from the fact that OBDD can be rewritten into structured d-DNNF that making such a d-DNNF complete may increase the width arbitrarily when one does not change the vtree~\cite{BolligW00}.

\paragraph{Using constants.} Our definition of complete structured DNNF does not allow constant inputs. This is in general not a problem as constants can be propagated in the circuits and thus eliminated. However, it is not directly clear how this propagation could affect the width in our setting. Moreover, most of our algorithms are easier to describe by allowing constants. So let us spend some time to deal with constants in our setting. To this end, we introduce the notion of {\em extended vtrees}. An extended vtree $T$ on a variable set $X$ is defined as a vtree in which we allow some leaves to be unlabeled. Every variable of $X$ must be the label of exactly one leaf still. A complete structured DNNF $(D,T,\lambda)$ is defined as for an extended vtree with the additional requirement that for every unlabeled leaf $\ell$ of $T$, $\lambda(\ell)$ is a set of constant inputs of $D$.

We now show that we can always remove the unlabeled leaves without increasing the width. 

\begin{lemma}
  \label{lem:removecstleaf}
  There is a linear time algorithm that, given a complete structured DNNF (resp. d-DNNF) $(D,T,\lambda)$ of width $w$ where $T$ is an extended vtree, computes a complete structured DNNF (resp. d-DNNF) $(D',T',\lambda')$ of width $w$ that is equivalent to $D$.
\end{lemma}
\begin{proof}
Given an extended vtree $T$ and a leaf $\ell$, let $T \setminus \ell$ be the vtree obtained by removing the leaf $\ell$ of $T$ and by merging the father and the sibling of $\ell$ in $T$.
We first show that there is an algorithm that, given a complete structured DNNF (resp. d-DNNF) $(D,T,\lambda)$ of width $w$ and a non-labeled leaf $\ell$ of $T$, computes in linear time an equivalent complete structured DNNF $(D',T \setminus \ell, \lambda')$ of width at most $w$.

Let $t$ be the father and $t_s$ the sibling of $\ell$ in $T$. We let $t'$ be the vertex of $T \setminus \ell$ obtained by merging $t$ and $t_s$. By definition, all gates of $\lambda(t)$ that are connected to gates in $\lambda(\ell)$ are $\wedge$-gates. We remove every $\wedge$-gate of $\lambda(t)$ connected to constant $0$ as they are equivalent to $0$ and are connected to $\vee$-gates of $\lambda(t)$. We next deal with the $\wedge$-gates of $\lambda(t)$ connected to the constant $1$. For every such gate $v$, we connect its other input to all output of $v$. This does not change the value computed by the output of $v$ and does not affect the determinism of the DNNF.

Now observe that the circuit has the following form: $\vee$-gates of $\lambda(t)$ are connected to $\vee$-gates of $\lambda(t_s)$. Without changing the function computed nor determinism, we can connect the $\vee$-gates of $\lambda(t)$ directly to the input of its inputs and thus remove every $\vee$-gate of $\lambda(t_s)$. Now the circuit has the following form: $\vee$-gates of $\lambda(t)$ are connected to $\wedge$-gates of $\lambda(t_s)$. We thus define $\lambda'(t')$ as the remaining $\vee$-gates of $\lambda(t)$ and $\wedge$-gates of $\lambda(t_s)$ and get a complete structured DNNF for $T \setminus \ell$. The number of $\vee$-gates in $\lambda(t')$  is less than in $\lambda(t)$ so the width has not increased.

Iterating this construction and observing that every $\lambda(t)$ is treated only once, we get the claim of the lemma.
\end{proof}

\subsection{Existential quantification on bounded width d-DNNF}
\label{sec:existquant}

In this section, we give an algorithm that allows us to quantify variables in d-DNNF. The main result is the following.

\begin{theorem}\label{thm:mainprojection}
There is an algorithm that, given a complete structured  DNNF $(D,T,\lambda)$ of width $w$ and $Z \subseteq \var(D)$, computes in time $2^{O(w)}|D|$ a complete structured d-DNNF $(D',T',\lambda')$ of width at most $2^w$ having a gate computing $\exists Z~D$ and another gate computing $\neg \exists Z~D$.
\end{theorem}

In the remainder of this section, we will prove Theorem~\ref{thm:mainprojection}. Let $(D,T,\lambda)$ be a complete structured DNNF. Let $X$ be the set of variables of $D$, $Z \subseteq X$ the variables that we will quantify and $w$ the width of~$D$.

Given a node $t$ of $T$, let $\var(t)$ be the set of variables which are at the leaves of the subtree of $T$ rooted in $t$. We define $\forgot(t) := Z \cap \var(t)$  and $\kept(t) := \var(t) \setminus \forgot(t)$. Intuitively, $\forgot(t)$ contains the the set of variables that are quantified away below $t$ while $\kept(t)$ contains the remaining variables under $t$. Let $D_v$ for a gate $v$ denote the sub-DNNF of $D$ rooted in $v$.

\paragraph{Shapes.} A key notion for our algorithm will be what we call {\em shapes}. Let $t$ be a node of $T$ and let $O_t$ be the set of $\vee$-gates of $D$ labelling $t$. An assignment $\tau : \kept(t)\rightarrow\{0,1\}$ is of shape $S \subseteq O_t$ if and only if \[ S = \{s \in O_t \mid \exists \sigma: \forgot(t) \rightarrow \{0,1\}, \tau \cup \sigma \models D_s\}.\]

We denote by $\shape_t \subseteq 2^{O_t}$ the set of shapes of a node $t$. Observe that $|\shape_t| \leq 2^{|O_t|} \leq 2^w$ since $|O_t|\leq w$ by definition.

The key observation is that $\shape_t$ can be inductively computed. Indeed, let $t$ be a node of $T$ with children $t_1,t_2$ and let $S_1 \in \shape_{t_1}$, $S_2 \in \shape_{t_2}$. We define $S_1 \join S_2 \subseteq O_t$ to be the set of gates $s \in O_t$ that evaluate to $1$ once we replace every gate in $S_1$ and $S_2$ by $1$ and every gate in $O_{t_1} \setminus S_1$ and $O_{t_2} \setminus S_2$ by $0$.


\begin{lemma}
  \label{lem:compshape}
Let $t$ be node of $T$ with children $t_1,t_2$. Let $\tau_1 : \kept(t_1) \rightarrow \{0,1\}$ be of shape $S_1$ and $\tau_2 : \kept(t_2) \rightarrow \{0,1\}$ be of shape $S_2$ be of shape $S_2$. Then $\tau = \tau_1 \cup \tau_2$ is of shape $S_1 \join S_2$.
\end{lemma}
\begin{proof}
  Let $S$ be the shape of $\tau$. We first prove $S \subseteq S_1 \join S_2$. So let $s \in S$. Since $\tau$ is of shape $S$, there exists $\sigma: \forgot(t) \rightarrow \{0,1\}$ such that $\tau \cup \sigma$ satisfies $D_s$. Since $s$ is a $\vee$-gate, there must be an input gate $s'$ of $s$ such that $\tau \cup \sigma$ satisfies $s'$. By definition, $s'$ is a $\wedge$-gate with two children $s_1 \in O_{t_1}$ and $s_2 \in O_{t_2}$. Thus $D_{s_1}$ is satisfied by $(\tau \cup \sigma)|_{\var(t_1)} = \tau_1 \cup \sigma|_{\var(t_1)}$. Consequently, $s_1 \in S_1$ since $S_1$ is the shape of $\tau_1$. Similarly $s_2 \in S_2$. Thus, in the construction of $S_1 \join S_2$, both $s_1$ and $s_2$ are replaced by $1$, so $s$ evaluates to $1$, that is, $s \in S_1 \join S_2$.
  
  We now show that $S_1 \join S_2 \subseteq S$. So let $s \in S_1 \join S_2$. Then, in the construction of $S_1 \join S_2$, there must be an input gate of $s$ that is satisfied. So there is an input $s'$ of $s$, that is a $\wedge$-gate with children $s_1 \in O_{t_1}$, $s_2 \in O_{t_2}$ evaluating to $1$. It follows that $s_1$ and $s_2$ have been replaced by $1$ in the construction of $S_1 \join S_2$. Now by definition of $S_1$, there exists $\sigma_1 \colon \forgot(t_1)\rightarrow\{0,1\}$ such that $\tau_1 \cup \sigma_1$ satisfies $D_{s_1}$ and $\sigma_2 \colon \forgot(t_2)\rightarrow\{0,1\}$ such that $\tau_2 \cup \sigma_2$ satisfies $D_{s_2}$. Thus, $(\tau_1 \cup \sigma_1) \cup (\tau_2 \cup \sigma_2) = \tau \cup (\sigma_1 \cup \sigma_2)$ is well-defined because $\sigma_1$ and $\sigma_2$ do not share any variables because $s'$ is decomposable. Moreover, $\tau \cup (\sigma_1 \cup \sigma_2)$ satisfies $D_s$ and thus we have $s \in S$.
  \end{proof}

  \paragraph{Constructing the projected d-DNNF.} We now inductively construct a d-DNNF $D'$ computing $\exists Z~D$ and of width at most $2^w$. The extended vtree $T'$ for $D'$ is obtained from $T$ by removing the labels of the leaves corresponding to variables in $Z$. One can then apply Lemma~\ref{lem:removecstleaf} to obtain a vtree. We inductively construct for every node $t$ of $T$ and $S \in \shape_t$, an $\vee$-gate $v_t(S)$ in $D'$ such that $D'_{v_t(S)}$ accepts exactly the assignment of shape $S$ and we will define $\lambda'(t) = \bigcup_{S \in \shape_t} v_t(S)$.

  If $t$ is a leaf of $T$, then $\kept(t)$ has at most one variable, thus we have at most two assignments of the form $\kept(t) \rightarrow \{0,1\}$. We can thus try all possible assignments to compute $\shape_t$ explicitly and $v_t(S)$ will either be a literal or a constant for each $S \in \shape_t$. We put $v_t(S)$ in $\lambda'(t')$ where $t'$ is the leaf of $T'$ corresponding to $t$. It is clear that if $t'$ is labeled with variable $x$ then $v_t(S)$ is a literal labeled by $x$ or by $\neg x$. If $t'$ is unlabeled, then it corresponds to a leaf $t$ of $T$ labeled with a variable of $Z$. Thus $v_t(S)$ is a constant input so the conditions of structuredness are respected.

Now let $t$ be a node of $T$ with children $t_1,t_2$ and assume that we have constructed $v_{t_1}(S_1)$ for every $S_1 \in \shape_{t_1}$ and $v_{t_2}(S_2)$ for every $S_2 \in \shape_{t_2}$. We define $v_t(S)$ as:
\[ \bigvee_{S_1, S_2: S = S_1\join S_2} v_{t_1}(S_1) \wedge v_{t_2}(S_2) \]
where $S_1,S_2$ run over $\shape_{t_1}$ and $\shape_{t_2}$ respectively.

First of all, observe that the $\wedge$-gates above are decomposable since $D'_{v_{t_1}(S_1)}$ is on variables $\kept(t_1)$ which is disjoint from $\kept(t_2)$, the variables of $D'_{v_{t_2}(S_2)}$.

Moreover, observe that the disjunction is deterministic. Indeed, by induction, $\tau$ satisfies the term $v_{t_1}(S_1) \wedge v_{t_2}(S_2)$ if and only if  $\tau|_{\var(t_1)}$ is of shape $S_1$ and $\tau|_{\var(t_2)}$ is of shape $S_2$. Since an assignment has exactly one shape, we know that $\tau$ cannot satisfy another term of the disjunction.

Finally, we have to show that $v_t(S)$ indeed computes the assignments of shape $S$. This is a consequence of Lemma~\ref{lem:compshape}. Indeed, if $\tau$ is of shape $S$ then let $S_1, S_2$ be the shapes of $\tau|_{\var(t_1)}$ and $\tau|_{\var(t_2)}$ respectively. By Lemma~\ref{lem:compshape}, $S=S_1 \join S_2$ and then $\tau \models v_{t_1}(S_1) \wedge v_{t_2}(S_2)$, and then, $\tau \models v_t(S)$.

Now, if $\tau \models v_{t_1}(S_1) \wedge v_{t_2}(S_2)$ for some $S_1$ and $S_2$ in the disjunction, then we have by induction that $\tau|_{\var(t_1)}$ and $\tau|_{\var(t_2)}$ are of shape $S_1$ and $S_2$ respectively. By Lemma~\ref{lem:compshape}, $\tau$ is of shape $S_1 \join S_2 = S$. 

Let $t'$ be the node of $T'$ corresponding to $t$. We put all gates needed to compute $v_t(S)$ in $\lambda'(t')$ for every $S$. This has the desired form: a level of $\vee$-gate, followed by a level of $\wedge$-gate connected to $\vee$-gates in $\lambda'(t'_1)$ and $\lambda'(t'_2)$.  By construction, the width of the d-DNNF constructed so far is  $\max_t |\shape_t| \leq 2^w$.

Now assume that we have a d-DNNF $D_0$ with a gate $v_t(S)$ for every $t$ and every $S \in \shape_t$ computing the assignments of shape $\tau$. Let $r$ be the root of $T$. We assume w.l.o.g.~that the root of $D$ is a single $\vee$-gate $r_o$ connected to every $\wedge$-gate labeled by $r$. Then $v_r(\{r_o\})$ accepts exactly $\exists Z D$ and $v_r({\emptyset})$ accepts $\neg \exists Z~D$.

\section{Algorithms for graph width measures}\label{sec:graphs}

In this section, we will show how we can use the result of Section~\ref{sec:ddnnf} in combination with known compilation algorithms to show tractability results for QBF with restricted underlying graph structure and bounded quantifier alternation. This generalizes the results of~\cite{Chen04,FerraraPV05,FichteMHW18}. 

We use the following result which can be verified by careful analysis of the construction in~\cite[Section 3]{Darwiche01}; for the convenience of the reader we give an independent proof in Appendix~\ref{app:treewidthlinear}.
\begin{theorem}\label{thm:compileptw}
There is an algorithm that, given a CNF $F$  of primal treewidth $k$, computes in time $2^{O(k)}|F|$ a complete structured d-DNNF $D$ of width $2^{O(k)}$ equivalent to $F$.
\end{theorem}

We lift Theorem~\ref{thm:compileptw} to incidence treewidth by using the following result from~\cite{LampisMM18}.

\begin{proposition}\label{prop:twtransfer}
 There is an algorithm that, given a CNF-formula $F$ of incidence treewidth $k$, computes in time $O(2^k |F|)$ a 3CNF-formula $F'$ of primal treewidth $O(k)$ and a subset $Z$ of variables such that $F\equiv \exists Z F'$.
\end{proposition}

\begin{corollary}\label{cor:compileitw}
There is an algorithm that, given a CNF $F$ formula of primal treewidth $k$, computes in time $2^{O(k)}|F|$ a complete structured d-DNNF $D$ of width $2^{O(k)}$ and a subset $Z$ of variables such that $F\equiv \exists Z D$.
\end{corollary}

Note that in~\cite{BovaCMS15} there is another algorithm that compiles bounded incidence treewidth into d-DNNF without introducing new variables that have to be projected away to get the original function. The disadvantage of this algorithm though is that the time to compile is quadratic in the size of $F$. Since we are mostly interested in QBF in which the last quantifier block is existential, adding some more existential variables does not hurt our approach much, so we opted for the linear time algorithm we get from Corollary~\ref{cor:compileitw}.

Now using Theorem~\ref{thm:mainprojection} iteratively, we directly get the following result.

\begin{theorem}\label{thm:itw}
 There is an algorithm that, given a QBF $F$ with free variables, $\ell$ quantifier blocks and of incidence treewidth $k$, computes in time $\exp^{\ell+1}(O(k))|F|$ a complete structured  d-DNNF of width $\exp^{\ell+1}(O(k))$ accepting exactly the models of $F$.
\end{theorem}
\begin{proof}
Let $F = Q_1X_1\dots \exists X_\ell G$. We use Corollary~\ref{cor:compileitw} to construct a structured DNNF $D$ of width $2^{O(k)}$ such that $G \equiv \exists Z D$, that is $F \equiv Q_1X_1\dots \exists (X_\ell \cup Z) G$. By projection $X_\ell \cup Z$ using Theorem~\ref{thm:mainprojection}, we can construct a complete structured  d-DNNF $D'$ of width $2^{2^{O(k)}}$ computing $\exists X_\ell G$ and $\neg \exists X_\ell G \equiv \forall X_\ell \neg G$ simultaneously. Now we apply iteratively Theorem~\ref{thm:mainprojection} on $X_i$ to compute simultaneously $Q_{k} X_k \dots \exists X_\ell G$ and $\neg Q_k X_k \neg (Q_{k+1} X_{k+1} \dots \exists X_\ell G)$. This is possible to maintain it inductively since $\forall X_k A \equiv \neg \exists X_k \neg A$. Each step blows the width of the circuit by a single exponential, resulting in the stated complexity.
\end{proof}

As an application of Theorem~\ref{thm:itw}, we give a result on model counting.

\begin{corollary}\label{cor:countingitw}
 There is an algorithm that, given a QBF $F$ with free variables, $\ell$ quantifier blocks and of incidence treewidth $k$, computes in time $\exp^{\ell+1}(O(k))|F|$ the number of models of $F$.
\end{corollary}

We remark that Corollary~\ref{cor:countingitw} generalizes several results from the literature. On the one hand, it generalizes the main result of~\cite{Chen04} from decision to counting, from primal treewidth to incidence treewidth and gives more concrete runtime bounds\footnote{We remark that the latter two points have already been made recently in~\cite{LampisMM18}.}. On the other hand, it generalizes the counting result of~\cite{FichteMHW18} from projected model counting, i.e., QBF formulas free variables and just one existential quantifier block, to any constant number of quantifier alternations. Moreover, our runtime is linear in~$|F|$ in contrast to the runtime of~\cite{FichteMHW18} which is quadratic.

As a generalization of Theorem~\ref{thm:itw}, let us remark that there are compilation algorithms for graph measures beyond treewidth. For example, it is known that CNF formulas of \emph{bounded signed cliquewidth}~\cite{FischerMR08} can be compiled efficiently~\cite{BovaCMS15}. More exactly, there is an algorithm that compiles a CNF formula $F$ of signed incidence cliquewidth $k$ in time $2^{O(k)} |F|^2$ into a structured d-DNNF of size $2^{O(k)} |F|$. We will not formally introduce signed incidence cliquewidth here but refer the reader to~\cite{FischerMR08,BraultCM15}. Inspecting the proof of~\cite{BovaCMS15}, one can observe that the algorithm construct a complete structured d-DNNF of width at most $2^{O(k)}$ which as above yields the following result.

\begin{theorem}\label{thm:scw}
 There is an algorithm that, given a QBF $F$ with free variables, with $\ell$ quantifier blocks and of signed incidence cliquewidth $k$, computes in time $\exp^{\ell+1}(O(k))|F|+ 2^{O(k)} |F|^2$ a complete structured  d-DNNF of width $\exp^{\ell+1}(O(k))$ accepting exactly the models of $F$.
\end{theorem}

With Theorem~\ref{thm:scw} it is now an easy exercise to derive generalizations of~\cite{Chen04,FichteMHW18,FischerMR08}.

In the light of the above positive results one may wonder if our approach can be pushed to more general graph width measures that have been studied for propositional satisfiability like for example modular treewidth~\cite{PaulusmaSS16} or (unsigned) cliquewidth~\cite{SlivovskyS13}. Using the results of~\cite{LampisM17}, we can answer this question negatively in two different ways: on the one hand, QBF of bounded modular cliquewidth and bounded incidence cliquewidth with one quantifier alternation is $\mathsf{NP}$-hard, so under standard assumptions there is no version of Theorem~\ref{thm:itw} and thus also not of Corollary~\ref{cor:compileitw} for cliquewidth. On the other hand, analyzing the proofs of~\cite{LampisM17}, one sees that in fact there it is shown that for every CNF formula $F$ there is a bounded modular treewidth and bounded incidence treewidth formula $F'$ and a set $Z$ of variables such that $F\equiv \exists Z F'$. Since it is known that there are CNF formulas that do not have subexponential size DNNFs~\cite{BovaCMS16}, it follows that there are such formulas $F'$ such that every DNNF representation of $\exists Z F'$ has exponential width. This unconditionally rules out a version of Corollary~\ref{cor:compileitw} and Theorem~\ref{thm:scw} for modular treewidth or cliquewidth.

\section{Transformations of bounded width d-DNNF}
\label{sec:transformations}

In this section, we systematically study the tractability of several transformations in the spirit of the knowledge compilation map of Darwiche and Marquis~\cite{DarwicheM02}. Given complete structured  d-DNNF $(D,T,\lambda)$ of width $w$  and $Z \subseteq \var(D)$, we will be interested in the following transformations:
\begin{itemize}
\item {\bf Conditioning (CD)}: given $\tau : Z \rightarrow \{0,1\}$, construct a complete structured  d-DNNF computing $D[\tau]$.
\item {\bf Forgetting (FO)}: construct a complete structured  d-DNNF computing $\exists Z~D$.
\item {\bf Negation ($\neg$)}: construct a complete structured  d-DNNF computing $\neg D$.
\item {\bf Bounded conjunction ($D \wedge D'$)}: given $D'$ a complete structured  d-DNNF with the same vtree $T$ as $D$, construct a d-DNNF computing $D \wedge D'$.
\item {\bf Conjunction ($\wedge$)}: given $D_1,\dots,D_n$ complete structured d-DNNF with the same vtree $T$ as $D$, construct a d-DNNF computing $D_1 \wedge \dots \wedge D_n$.
\item {\bf Bounded disjunction ($D \vee D'$)}: given $D'$ a complete structured d-DNNF with the same vtree $T$ as $D$, construct a d-DNNF computing $D \vee D'$.
\item {\bf Conjunction ($\vee$)}: given $D_1,\dots,D_n$ complete structured d-DNNF with the same vtree $T$ as $D$, construct a d-DNNF computing $D_1 \vee \dots \vee D_n$.
\end{itemize}
The tractability of these transformations is summarized in Table~\ref{tab:kcmap}.

\begin{table*}
  \centering
  \begin{tabular}{|c|c|c|}
    \hline
    \bf Transformation & \bf Width & \bf Proof \\ \hline
    \bf CD & $\leq w$ & Lemma~\ref{lem:removecstleaf} \\ \hline
    \bf FO & $\leq 2^w$ & Theorem~\ref{thm:mainprojection} \\ \hline
    \bf $\neg D$ & $\leq 2^w, \geq 2^{\Omega(w)}$ & Theorem~\ref{thm:mainprojection} applied with $Z = \emptyset$ and Theorem~\ref{thm:unboundedop} \\ \hline 
    \bf $D \wedge D'$ & $\leq ww'$ & \cite{PipatsrisawatD08} (see Appendix~\ref{thm:productwidth}, Theorem~\ref{thm:productwidth})\\ \hline
    \bf $\bigwedge_{i=1}^n D_i$ & Unbounded & Theorem~\ref{thm:unboundedop} \\ \hline
    \bf $D \vee D'$ & $\leq 2^{w+w'}$ & Theorem~\ref{thm:productwidth} \\ \hline
    \bf $\bigvee_{i=1}^n D_i$ & Unbounded & Theorem~\ref{thm:unboundedop}   \\ \hline
  \end{tabular}
  \caption{Transformations of bounded width d-DNNF}
  \label{tab:kcmap}
\end{table*}

The tractability of bounded conjunction for structured d-DNNF was proven in~\cite{PipatsrisawatD08} but the upper bound is the product of the size of the inputs and not the product of the width. The construction is a product of each gate which makes it easy to see that the width of the resulting circuit is the product of the widths of the inputs. For the convenience of the reader we sketch here the construction of the conjunction of two complete structured DNNF (the full proof is given in Appendix~\ref{app:productwidth}). 

\begin{theorem}
  \label{thm:productwidth} Let $T$ be a vtree, $(D,T,\lambda)$ and $(D',T,\lambda')$ be two complete structured d-DNNF of width $w$ and $w'$. There exists a complete structured d-DNNF $(D'',T,\lambda'')$ of width $ww'$ computing $D \wedge D'$.
\end{theorem}
\begin{proof}[Proof (sketch).]
The construction is by induction on $T$: for every $t$ and for every $u \in \lambda(t)$ and $u' \in \lambda(t')$ of the same type ($u$ and $u'$ are either both $\vee$-gates or $\wedge$-gates), we construct a gate $g_t(u,u')$ in $\lambda''(t)$ computing $D_u \wedge D_{u'}$. The construction is straightforward when $t$ is a leaf. If $t$ is an internal node with children $t_1,t_2$ we have two cases. If $u$ and $u'$ are $\wedge$-gates, we define $g_t(u,u') = g_{t_1}(u_1,u_1') \wedge  g_{t_2}(u_2,u_2')$ where $u_1,u_2$ are the children of $u$ and $u_1',u_2'$ the children of $u'$. If $u$ and $u'$ are $\vee$-gates, we define $g_t(u,u') = \bigvee_{i,j} g_t(u_i,u'_j)$ where $u_i$ are the children of $u$ and $u_i'$ the children of $u'$. The full proof is given in Appendix~\ref{app:productwidth}. 
\end{proof}

The following theorem proves the optimality of our result concerning the negation of complete structured DNNF and show that the width may blow up when one computes the unbounded conjunction or disjunction of small width complete structured DNNF.

\begin{theorem}
  \label{thm:unboundedop} For every $n$, there exist complete structured d-DNNF $D_1, \dots,$ $D_{n}$ on variables $X$ with $|X|=O(n)$ having the same vtree $T$ of width $2$ such that $\bigwedge_{i=1}^{n} D_i$ cannot be represented by complete structured DNNF of width smaller than $2^{\Omega(n)}$. Moreover  $\bigvee_{i=1}^{n} \neg D_i$ can be represented by complete structured DNNF of width $\Omega(n)$ but not less.
\end{theorem}
\begin{proof}
  Let $C$ be a conjunction of literals on $X$. It is easy to see that for any vtree $T$ on $X$, $C$ can be computed by a complete structured d-DNNF with vtree $T$ of width $1$. Indeed, $C$ is equivalent to $C \wedge \bigwedge_{x \notin \var(C)} x \vee \neg x$ and one can reorder this decomposable conjunction and use associativity to mimic $T$. It is easy to see that the clause $\neg C$ can be computed by a complete structured d-DNNF with vtree $T$ of width $2$ (either by constructing it explicitly or by applying Theorem~\ref{thm:mainprojection}).

  Let $F = \bigwedge_{i=1}^{n} D_i$ be the CNF formula from~\cite{BovaCMS16} on variables $X$ that cannot be represented by DNNF of size smaller than $2^{\Omega(n)}$. From what precedes, $F$ is the unbounded conjunction of width $2$ d-DNNF that cannot be represented by complete structured DNNF of width smaller than $2^{\Omega(n)}$.

  Now, assume that $\neg F = \bigvee_{i=1}^n \neg D_i$ can be represented by a width $w$ complete structured DNNF. Then by Theorem~\ref{thm:mainprojection}, $F$ is computed by a complete structured DNNF of width $2^w \geq 2^{\Omega(n)}$ from what precedes, that is $w \geq \Omega(n)$. Moreover, observe that $F$ is a DNF with $n$ terms. It can thus be easily computed by a complete structured DNNF of width $n$ by having a $\vee$-gate on top of all its terms, represented by width $1$ DNNF as described in the beginning of this proof, resulting in a DNNF of width $n$.
\end{proof}

\section{Lower Bounds}
\label{sec:lower-bounds}

In this section, we will show that all restrictions we put onto the DNNF in Theorem~\ref{thm:mainprojection} are
necessary. 

\subsection{The definition of width}

Width of an OBDD is usually defined on \emph{complete} OBDD\footnote{We remark that this is similar for more general representations like structured DNNF, but we will not follow this direction here.}. There is however another way of defining width for OBDD by just counting the number of nodes that are labeled with the same variable. Let us call this notion \emph{weak width}. We will show that width in Theorem~\ref{thm:mainprojection} cannot be substituted by weak width.

\begin{lemma}
 For every $n$ there is an OBDD $D_n$ in $O(n)$ variables of weak width $3$ and a subset $Z$ of such that $\neg \exists Z~D_n$ does not have an OBDD of size $2^{o(n)}$.
\end{lemma}
\begin{proof}
 Let $S_i$ for $i\in \mathbb{N}$ denote the term $ \neg z_i \wedge \big(\bigwedge_{j\in [i-1]} z_j\big )$. For a CNF $F= C_1\land \ldots \land C_m$ we then define the function 
 \[F' = \bigvee_{i=1}^m S_i \land C_i.\]
 It is easy to see that by testing $z_1,\dots,z_m$ successively and branching a small OBDD for $C_i$ at each $0$-output of the decision node testing $z_i$ as depicted on Figure~\ref{fig:wwobdd}, one can construct an OBDD of size $O(|F|)$ computing $F'$. If every variable appears in at most three clauses of $F$, then this OBDD has weak width $3$ since a variable $x$ is only tested for clauses where it appears.

 \begin{figure}
   \centering
   \begin{tikzpicture}[node distance=1.5cm]
     \node (z1) {$z_1$};
     \node[right of=z1] (z2) {$z_2$};
     \node[right of=z2] (dots) {$\dots$};
     \node[right of=dots] (zm) {$z_m$};

     \draw[->] (z1) -- (z2);
     \draw[->] (z2) -- (dots);
     \draw[->] (dots) -- (zm);

     \node[above of=z1] (C1) {$C_1$};
     \node[above of=z2] (C2) {$C_2$};
     \node[above of=zm] (Cm) {$C_m$};

     \draw[dashed, ->] (z1) -- (C1);
     \draw[dashed, ->] (z2) -- (C2);
     \draw[dashed, ->] (zm) -- (Cm);

   \end{tikzpicture}
   \caption{An OBDD for $F'$}
   \label{fig:wwobdd}
 \end{figure}
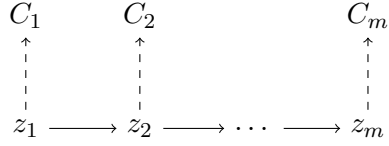
 
 Note that $\forall Z~F' \equiv F$. Since there are CNF formulas of the desired type that do not have subexponential size DNNF, it follows that for such $F$ the function $\forall Z~F'$ has exponential size. Now remarking that $\forall Z~F'\equiv \neg \exists Z~(\neg F')$ and that $\neg F'$ has an OBDD of weak width $3$ as well, completes the proof. 
\end{proof}

\subsection{Structuredness}

One of the properties required for Theorem~\ref{thm:mainprojection} is that we need the input to be structured. Since structuredness is quite restrictive, see e.g.~\cite{PipatsrisawatD10}, it would be preferable to get rid of it to show similar results. Unfortunately, there is no such result as the following lemma show.

To formulate our results, we need a definition of width for FBDD. This is because width as we have defined it before depends on the vtree of the DNNF which we do not have in the case without structuredness. To define width for the unstructured case, we consider layered FBDD: an FBDD $F$ is called \emph{layered} if the nodes of $F$ can be partitioned into sets $L_1, \ldots, L_s$ such that for every edge $uv$ in $F$ there is an $i\in [s]$ such that $u\in L_i$ and $v\in L_{i+1}$. The \emph{width} of $F$ is then defined as $\max\{|L_i| \mid i\in [s]\}$.

\begin{lemma}
 For every $n$ there is a function $f_n$ in $O(n^2)$ variables with an FBDD representation of size $O(n^2)$ and width $O(1)$ such that there is a variable $x$ of $f_n$ such that every deterministic DNNF for $\exists x~f_n$ has size $2^{\Omega(n)}$.
\end{lemma}
\begin{proof}
 We use a function introduced by Sauerhoff~\cite{Sauerhoff03}: let $g:\{0,1\}^n\rightarrow \{0,1\}$ be the function that evaluates to $1$ if and only if the sum of its inputs is divisible by $3$. For a $n\times n$-matrix $X$ with inputs $x_{ij} \in \{0,1\}$, we define \[R_n(X) := \bigoplus_{i=1}^n g(x_{i1}, x_{i2}, \ldots, x_{in})\] where $\oplus$ denotes addition modulo $2$ and define $C_n(X) := R_n(X^T)$ where $X^T$ is the transpose of $X$. Then $S_n(X) := R_n(X) \lor C_n(X)$.
 
 Note that, ordering the variables of $X$ by rows, resp.~columns, $R_n$ and $C_n$ both have OBDD of width $O(1)$ and size $O(n^2)$. Now let $S_n' = (x \land R_n) \lor (\neg x \land C_n)$. Then $S_n'$ clearly has an FBDD of size $O(n^2)$ and width $O(1)$: decide on $x$ first and then depending on its value follow the OBDD for $R_n$ or $C_n$.
 
 But $\exists x S_n'(X) = S_n(X)$ which completes the proof since $S_n$ is known to require size $2^{\Omega(n)}$ for deterministic DNNF~\cite{BovaCMS16}.
\end{proof}


%
%



\bibliographystyle{alpha}
\bibliography{qbfcompile}

\appendix

\section{Compiling bounded primal treewidth CNF}\label{app:treewidthlinear}
\label{sec:compiling}

A tree decomposition $(T, (B_t)_{t\in T})$ of a graph $G$ consists of a tree $T$ and a set of \emph{bags} $B_t\subseteq V(G)$ such that for every node $t$ of the tree there is exactly one bag and the following properties hold: (i) for every edge $e\in E(G)$, there is a node $t$ of $T$ such that $e\subseteq B_t$; (ii) for every vertex $v\in V(G)$, the set $\{t\in V(T)\mid v\in B_t\}$ induces a subtree of $T$. The \emph{width} of a decomposition is $\max\{|B_t| -1\mid t\in V(T)\}$ and the \emph{treewidth} of $G$ is the smallest width of a tree decomposition of $G$.
The \emph{primal treewidth} of a CNF formula is the treewidth of its primal graph.

This section is dedicated to the proof of the following theorem:
\begin{theorem}
  \label{thm:ptwcompilelin}
  There is an algorithm that, given a CNF formula $F$ of primal treewidth $k$, constructs in time $2^{O(k)}|F|$ a complete structured decision DNNF of size $2^{O(k)}|F|$ equivalent to $F$.
\end{theorem}

In the remainder, we fix a formula $F$ and compute the DNNF in several steps.

\subsection{Linear time computation of a nice tree decomposition}

To simplify the proof, we will work with \emph{nice} tree decompositions~\cite{Kloks94}: a tree decomposition $(T, (B_t)_{t\in T})$ of a graph $G$ is called \emph{nice} if all internal nodes $t$ of $T$ are of one of the following types:
\begin{itemize}
 \item \textbf{Introduce node:} $t$ has a single child $t'$ and there is a vertex $v\in V(G)\setminus B_{t'}$ such that $B_t = B_{t'} \cup \{v\}$.
 \item \textbf{Forget node:} $t$ has a single child $t'$ and there is a vertex $v\in B_{t'}$ such that $B_t = B_{t'}\setminus \{v\}$.
 \item \textbf{Join node:} $t$ has exactly two children $t_1$ and $t_2$ and we have $B_{t} = B_{t_1} = B_{t_2}$.
\end{itemize}

The first step of our algorithm is to compute a nice tree decomposition of width $O(k)$ for the primal graph of $F$ in time $2^{O(k)}|F|$. To this end, we use the algorithm of~\cite{bodlaender2016c} which, given a graph $G$ of treewidth $k$, computes in time $2^{O(k)}|G|$ a tree decomposition of $G$ of width at most $5\cdot k$. Note that the primal graph of $F$ has size at most $k |F|$, so the runtime for the computation of the tree decomposition is linear in $|F|$. This tree decomposition is then turned into a nice tree decomposition in linear time and without increasing the treewidth by standard techniques~\cite{Kloks94}. Denote the resulting nice tree decomposition of the primal graph of $F$ by $(T, (X_t)_{t\in T})$. W.l.o.g.~we assume that for all leaves $t$ we have $X_t= \emptyset$.

\subsection{Constructing the decision DNNF}
\label{sec:construction}

We start by describing how we construct the decision DNNF equivalent to~$F$. We first introduce some notation. Let $r$ be the root of $T$. We assume w.l.o.g.~that $X_r= \emptyset$. Let $T_t$ denote for every node $t$ of $T$ the subtree of $T$ rooted in $t$. For every clause $C$ of $F$, the variables in $C$ form a clique in the primal graph and thus there is a node $t$ in $T$ such that $\var(C) \subseteq X_t$. We denote by $t_C$ the node  $t$ of $T$ that is closest to the root $r$ such that $\var(C) \subseteq X_t$. Given a node $t$ of $T$, we denote by $\calC_t = \{C \in F \mid t_C = t\}$ and by $F_t = \bigcup_{u \in T_t} \calC_u$. Observe that $F_r = F$. 

Our construction proceeds by bottom-up dynamic programming on $T$ from the leaves to the root $r$. We construct a decision DNNF $D$ such that for every node $t$ of $T$ and $\tau : X_t \rightarrow \{0,1\}$, there exists a gate $v_t^\tau$ in $D$ computing $F_t[\tau]$. Observe that this is enough to prove Theorem~\ref{thm:ptwcompilelin} since we assume that $X_r = \emptyset$ and thus, there exists a gate in $D$ computing $F_r = F$. 

$D$ will be a complete structured d-DNNF for the labeled extended vtree $(T',\lambda)$ defined as follows: $T'$ has the same node as $T$ plus for every variable $x$, the only forget node on variable $x$ is connected to an extra leaf $t_x$ labeled by $x$ and every introduce node is connected to an unlabeled leaf. One can apply Lemma~\ref{lem:removecstleaf} to eliminate constants and have a regular vtree. In the following, we will always identify the vertices of $T$ with their corresponding vertices in $T'$. We assume by induction that $v_t^\tau$ is either a $\vee$-gate or an input that is in $\lambda(t)$.
Let $t$ first be a leaf. Since we assumed $X_{t_0} = \emptyset$, we have that $F_{t_0}$ is the empty CNF-formula and thus is by definition equivalent to the constant $1$. Thus, we add a gate $v_{t_0}^\emptyset := 1$.  We add $v_{t_0}^\emptyset$ to $\lambda(t_0)$. This obviously respects the condition on complete structured DNNF since $t_0$ is an unlabeled leaf of $T$.

\paragraph{$t$ is join node.} Let $t_1,t_2$ be the children of $t$. Observe that we have already constructed gates for $t_1$ and $t_2$ in $D$. Now, by definition, we have $X_t = X_{t_1} = X_{t_2}$ and $F_t = F_{t_1} \uplus F_{t_2} \uplus \calC_t$. Let $\tau \colon X_t \rightarrow \{0,1\}$. We start by evaluating $\calC_t$ on $\tau$. Recall that if $C \in \calC_t$, we have $\var(C) \subseteq X_t$ so $\tau$ assigns all variables that appear in $\calC_t$. Thus $\calC_t[\tau]$ is a constant. If $\calC_t[\tau] = 0$, then we add a new gate $v_t^\tau := 0$. Otherwise, $F_t[\tau]$ is equivalent to $F_{t_1}[\tau] \wedge F_{t_2}[\tau]$. We thus introduce a gate fan-in one $\vee$-gate $v_t^\tau$ and connect it to a $\wedge$-gate connected  to $v_{t_1}^\tau$ and $v_{t_2}^\tau$\footnote{The fan-in one $\vee$-gate is only necessary to respect the normal form of complete structured d-DNNF.}. By induction, $v_t^\tau$ computes $F_t[\tau]$. We still have to show that this new $\wedge$-gate is decomposable. This follows by the definition of tree decompositions: if $x$ is a variable that appears both in $F_{t_1}$ and $F_{t_2}$, then $x$ has to appear in $X_{t_1'}$ and $X_{t_2'}$ for $t_1'\in T_{t_1}$ and $t_2\in T_{t_2}$ and thus we get $x\in X_t$. It follows that $x$ is assigned a value by $\tau$ and so $x$ does not appear in the subcircuits rooted in $v_{t_1}^\tau$ and $v_{t_2}^\tau$.

We add both $v_t^\tau$ and the newly introduced $\wedge$-gate into $\lambda(t)$. By induction, $v_{t_1}^\tau$ and $v_{t_2}^\tau$ are $\vee$-gates or inputs in $\lambda(t_1)$ and $\lambda(t_2)$ respectively. The construction of $v_t^\tau$ thus respects the condition of complete structured DNNF.

\paragraph{$t$ is an introduce node.} Let $t_1$ be the child of $t$ and $x$ be the introduced variable. By definition, $X_t = X_{t_1} \cup \{x\}$ and $F_t = F_{t_1} \cup \calC_t$. Let $\tau \colon X_t \rightarrow \{0,1\}$ and let $\tau_1 := \tau|_{X_{t_1}}$. As in the previous case, we evaluate $\calC_t$ on $\tau$. If $\calC_t[\tau] = 0$, we proceed as before. Otherwise, $F_t[\tau] = F_{t_1}[\tau_1]$ and thus, we already have the gate $v_{t_1}^{\tau_1}$ that computes $F_t[\tau]$ and we let $v_t^\tau$ to be a fan in one $\vee$-gate connected to a fan in one $\wedge$-gates connected to $v_{t_1}^{\tau_1}$. We add all newly introduced gate to $\lambda(t)$. Since by induction, $v_{t_1}^{\tau_1}$ is in $\lambda(t_1)$, the construction of $v_t^\tau$ respects the condition of complete structured DNNF.

\paragraph{$t$ is a forget node.} Let $t_1$ be the child of $t$ and $x$ be the eliminated variable. By definition, $X_t = X_{t_1} \setminus \{x\}$ and $F_t = F_{t_1} \cup \calC_t$. Let $\tau \colon X_t \rightarrow \{0,1\}$ and let $\tau_0 := \tau \cup \{x \mapsto 0\}$ and $\tau_1 := \tau \cup \{x \mapsto 1\}$. As in the previous case, we evaluate $\calC_t$ on $\tau$. If $\calC_t[\tau] = 0$, we proceed as before. Otherwise we have that $F_t[\tau]$ is equivalent to $(x \wedge F_{t_1}[\tau_1]) \vee (\neg x \wedge F_{t_1}[\tau_0])$. We thus introduce a gate $v_t^\tau$ that is a decision node on variable $x$. We connect this gadget to $v_{t_1}^{\tau_1}$ and $v_{t_1}^{\tau_0}$ in the obvious way. By induction, $v_t^\tau$  computes $F_t[\tau]$.

Observe that $v_t^\tau$ is an $\vee$-gate connected to two $\wedge$-gate $w_1,w_2$. We put all of these gates in $\lambda(t)$ and the newly introduced input $x$ and $\neg x$ in $\lambda(t_x)$, where $t_x$ is the extra $x$-labeled leaf of $T'$. Since by induction $v_{t_1}^\tau$ is in $\lambda(t_1)$ for any $\tau$, the construction of $v_t^\tau$ respects the condition of complete structured DNNF.

\paragraph{Complexity of constructing $D$.} We now justify that this construction can be done in time $2^{O(k)}|F|$. For now, we assume that $\calC_t$ has been precomputed for every $t$. Computing these sets in linear time is not completely obvious and is thus done in the next section. The first thing that we need to do is to compute the order in which the nodes in $T$ are treated in the construction. This can be easily done in time $O(|T|) = O(|F|)$ by doing a depth-first search of $T$, starting from the root $r$. 

Now, we justify that the gates $v_t^\tau$ can be computed in time $2^{O(k)} |F|$ overall.
Observe that each case of the construction always boils down to the following steps: evaluate $\calC_t$ on $\tau : X_t \rightarrow \{0,1\}$ and branch it to existing gates in the circuit.

First, observe that every clause $C$ of $F$ is in exactly one $\calC_t$ so we have to evaluate $C$ at most $2^{O(k)}$ times in the algorithm. So the overall evaluation time for all $\calC_t$ is $2^{O(k)}|F|$.

Now, we show that we can compute the newly introduced gates from the existing ones in constant time. Observe that when we introduce $v_t^\tau$, either we label it with a constant, which can obviously be done in constant time or we connect it to at most two gates of the form $v_u^{\sigma}$ for $u$ a child of $t$ and $\sigma \colon X_u \rightarrow \{0,1\}$. To do this in constant time, we associate to each node $u$ of $T$ an array $A_u$ of size at most $2^{O(k)}$. Each entry of $A_u$ corresponds to an assignment $\sigma \colon X_u \rightarrow \{0,1\}$ (ordered by lexicographical order) and contains a pointer to the gate $v_u^\sigma$ in $D$. Now, when we create $v^\tau_t$, we create the gate in time $O(1)$ and insert a pointer to it in $A_t[\tau]$ in time $O(1)$. Then, we find the pointers to the appropriate gate by accessing in $O(1)$ time $A_u[\sigma]$ for $u$ a child of $t$.

\paragraph{Width of $D$.} It is easy to see that in our construction, the only $\vee$-gates contained in $\lambda(t)$ are the gates of the form $v_t^\tau$ for every node $t$ of $T$. Thus, the width of $D$ is at most $2^{O(k)}$.

\subsection{Computation of $\calC_t$ in linear time}

The last thing that we have to explain is how we compute $\calC_t$ in linear time. This is done using a bottom-up induction on $T$ and an appropriate data structure.

We start by observing that if $\calC_t \neq \emptyset$ then the father $u$ of $t$ in $T$ is a forget node. Indeed, if $u$ is not a forget node, then we have $X_u \supseteq X_t$. So if $\var(C) \subseteq X_t$, we also have $\var(C) \subseteq X_u$, and $u$ is closer to $r$ than $t$, we have $C \notin \calC_t$.

Now we claim that if $u$ is a forget node for variable $x$ then $\calC_t$ is exactly the set of clauses $C$ such that $x \in \var(C)$  and for every $t' \in T_t$ with $t'\ne t$, we have $C \notin \calC_{t'}$. To see this, first assume that $C\in \calC_t$. Then $C$ must contain $x$ with the same argument as above. Moreover, $C \notin \calC_{t'}$ for any node $t'\ne t$ since $t_C$ is uniquely defined. For the other direction, assume that $x \in \var(C)$  and for every $t' \in T_t$ with $t'\ne t$, we have $C \notin \calC_{t'}$. Since $x \notin X_u$, there is no $v$ outside of $T_t$ such that $x \in X_v$. So $t_C$ must lie in $T_t$. Moreover, since $C \notin \calC_{t'}$ for $t'$ for any $t'\in T$ with $t \ne t'$, we get $t_C = t$ or equivalently $C\in \calC_t$.

In order to compute $\calC_t$ for every $t$, we can thus compute them for every child of a forget node along a post-order depth-first traversal of $T$. When we computes $\calC_t$ for a node $t$ whose father is a forget $x$ node, we only have to add to $\calC_t$ every clause that contains $x$ and that we still have not added to a $\calC_{t_1}$ set, for $t_1 \prec t$. This can be done in linear time assuming a data structure $\mathcal{D}$ that allows us to remove from $\mathcal{D}$ every clause containing a variable $x$ in time $O(\sum_{C \in D \colon x \in \var(C)} |C|)$. Indeed, since every clause is deleted exactly once in this algorithm, the overall complexity of the algorithm would be $\sum_{C \in F} O(|C|) = O(|F|)$.

We now describe this data structure. For convenience, we assume that the variables of $F$ are $x_1, \dots, x_n$ and the clauses of $F$ are $C_1,\dots, C_m$. We start by initiating the following data structure: we have an array $A_V$ with $n$ entries and an array $A_F$ of size $m$. Each entry of $A_V$ is a doubly linked list. For $i \leq n$, $A_V[i]$ contains a doubly linked list of the clauses containing variable $x_i$. Moreover, $A_F$ contains for every clause a doubly linked list containing two way pointers between $A_F[j]$ and every occurrence of a clause $C_j$ in $A_V$. Figure~\ref{fig:datastructure} depicts the data structure for $(x_1 \vee x_2) \wedge (x_1 \vee x_3) \wedge (x_2 \vee x_3)$.

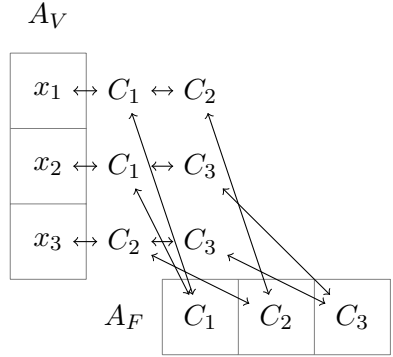
\begin{figure}
  \centering

  \begin{tikzpicture}
\draw[step=1cm,gray,very thin] (0,0) grid (1,-3);
\node (x1) at (0.5,-0.5) {$x_1$};
\node[below of=x1] (x2) {$x_2$};
\node[below of=x2] (x3) {$x_3$};

\node[right of=x1] (C1x1) {$C_1$};
\node[right of=C1x1] (C2x1) {$C_2$};

\node[right of=x2] (C1x2) {$C_1$};
\node[right of=C1x2] (C3x2) {$C_3$};

\node[right of=x3] (C2x3) {$C_2$};
\node[right of=C2x3] (C3x3) {$C_3$};

\draw[step=1cm,gray,very thin] (2,-3) grid (5,-4);

\node (C1) at (2.5,-3.5) {$C_1$};
\node[right of=C1] (C2) {$C_2$};
\node[right of=C2] (C3) {$C_3$};

\node[above of=x1] (AV) {$A_V$};
\node[left of=C1] (AF) {$A_F$};

\draw[<->] (x1) -- (C1x1); 
\draw[<->] (C1x1) -- (C2x1);

\draw[<->] (x2) -- (C1x2); 
\draw[<->] (C1x2) -- (C3x2);

\draw[<->] (x3) -- (C2x3); 
\draw[<->] (C2x3) -- (C3x3);

\draw[<->] (C1x1) -- (C1);
\draw[<->] (C1x2) -- (C1);
\draw[<->] (C2x1) -- (C2);
\draw[<->] (C2x3) -- (C2);
\draw[<->] (C3x2) -- (C3);
\draw[<->] (C3x3) -- (C3);

\end{tikzpicture}
  
  \caption{The data structure}
  \label{fig:datastructure}
\end{figure}

It is easy to see that such a data structure can be constructed in time $O(|F|)$ by simply reading the clauses of $F$ one after the other. When $x_i$ is seen in a clause $C_j$, we add $C_j$ in $A_V[i]$ together with a double link to the list in $A_F[j]$ and move to the next variable in $C_j$. 

The key property of this data structure is that given a variable $x_i$, we can remove every clause containing $x_i$ in the data structure in time $\sum_C |C|$ where the sum goes over these clauses. Indeed, we start by removing the first clause $C_j$ in the doubly linked list $A_V[i]$. This can be done in time $O(|C_j|)$ since we have to remove $|C_j|$ occurrences of $C_j$ and each of them can be found in $O(1)$ by reading $A_F[j]$ back and forth. We can apply this for every clause in $A_v[i]$ until the list is empty.

\section{Transformation of structured d-DNNF}
\label{app:productwidth}

This section is dedicated to the proof of Theorem~\ref{thm:productwidth}.

The construction is by induction on $T$. For every $t$, we construct in $\lambda''(t)$  a gate $g_t(u,u')$ for every $u \in \lambda(t)$ and $u' \in \lambda'(t)$ computing $D_u \wedge D_{u'}$, where $u$ and $u'$ are of the same type (input, $\wedge$-gate or $\vee$-gate).

Observe that $D''$ will be of width $ww'$ and if $r$ is the root of $T$, $u$ the output of $D$ and $u'$ the output of $D'$, then the gate $g_r(u,u')$ computes $D \wedge D'$.

We construct the gates $g_t(u,u')$ by induction on $t$. If $t$ is a leaf labeled by $x$ then $u$ and $u'$ are both inputs and $D_u \wedge D_{u'}$ is equivalent to either a constant, $x$ or $\neg x$.

Now let $t$ be a node of $T$ with children $t_1,t_2$ and let $u \in \lambda(t)$ and $u' \in \lambda'(t)$. First, assume that $u$ and $u'$ are $\wedge$-gates of $D$ and $D'$. Let $u_1 \in \lambda(t_1), u_2 \in \lambda(t_2)$ be the inputs of $u$ and $u'_1 \in \lambda'(t_1), u'_2 \in \lambda'(t_2)$ be the inputs of $u'$. We want $g_t(u,u')$ to compute $(u_1 \wedge u_2) \wedge (u'_1 \wedge u'_2)$. By associativity and commutativity, this is equivalent to $(u_1 \wedge u'_1) \wedge (u_2 \wedge u'_2)$ and by induction $(u_1 \wedge u'_1)$ is computed by gate $g_{t_1}(u_1,u_1')$ and $(u_2 \wedge u'_2)$ is computed by gate $g_{t_2}(u_2,u_2')$. Thus, we define $g_t(u,u')$ to be a decomposable $\wedge$-gate with input $g_{t_1}(u_1,u_1')$ and $g_{t_2}(u_2,u_2')$.

Now, assume that $u$ and $u'$ are $\vee$-gates of $D$ and $D'$. Let $u_1,\dots,u_k \in \lambda(t)$ be the inputs of $u$ and $u'_1,\dots,u'_p \in \lambda'(t)$ be the inputs of $u'$. We want $g_t(u,u')$ to compute $\bigvee_{i=1}^k u_i \wedge \bigvee_{j=1}^p u'_j = \bigvee_{i,j} u_i \wedge u'_j = \bigvee_{i,j} g_t(u_i,u'_j)$. Thus, we define $g_t(u,u')$ to be $\bigvee_{i,j} g_t(u_i,u'_j)$. It computes what we need but we have to check that the disjunction is deterministic. Assume that this is not the case, that is, there are $(i,j) \neq (k,l)$ such that $g_t(u_i,u'_j)$ and $g_t(u_k,u'_l)$ have a common satisfying assignment $\tau$. Then by definition, $\tau$ satisfies $(u_i \wedge u'_j) \wedge (u_k \wedge u'_l)$. We assume wlog that $i \neq k$. Then $\tau$ satisfies $u_i \vee u_k$ which contradicts the fact that $u$ is deterministic. 

\end{document}